\documentclass[11pt]{article}
\usepackage{amssymb}

\usepackage{graphicx}
\usepackage{amsmath}
\usepackage{makeidx}
\usepackage{indentfirst}

\newcounter{resultnum}[section]

\newcounter{conclusionnum}[section]

\newcounter{conditionnum}[section]

\newcounter{conjecturenum}[section]
\newtheorem{example}{Example}[section]

\newcounter{examplenum}[section]

\newcounter{exercisenum}[section]

\newcounter{lemmanum}[section]
\newtheorem{notation}{Notation}[section]

\newcounter{notationnum}[section]
\newtheorem{theorem}{Theorem}[section]

\newcounter{theoremnum}[section]
\newtheorem{definition}{Definition}[section]

\newcounter{definitionnum}[section]
\newtheorem{hypothesis}{Hypothesis}[section]

\newcounter{hypothesisnum}[section]

\newcounter{corollarynum}[section]

\newcounter{remarknum}[section]

\newcounter{propositionnum}[section]

\newcounter{acknowledgementnum}[section]

\newcounter{algorithmnum}[section]

\newcounter{axiomnum}[section]

\newcounter{casenum}[section]
\newtheorem{claim}{Claim}[section]

\newcounter{claimnum}[section]

\newcounter{summarynum}[section]

\newcounter{problemnum}[section]
\newenvironment{proof}[1][]{\textbf{Proof.} }{}

\begin{document}

\title{ Diffusion and Self--Organized Criticality in Ricci Flow Evolution of
Einstein and Finsler Spaces}
\date{October 10, 2010}
\author{ Sergiu I. Vacaru\thanks{
sergiu.vacaru@uaic.ro, Sergiu.Vacaru@gmail.com;\newline
http://www.scribd.com/people/view/1455460-sergiu } \\
{\quad} \\
{\small {\textsl{\ Science Department, University "Al. I. Cuza" Ia\c si},} }%
\\
{\small {\textsl{\ 54 Lascar Catargi street, 700107, Ia\c si, Romania}} }}
\maketitle

\begin{abstract}
Imposing non--integrable constraints on Ricci flows of (pseudo) Riemannian metrics we model mutual transforms to, and from, non--Riemannian spaces. Such evolutions of geometries and physical theories can be modelled for nonholonomic manifolds and vector/ tangent bundles enabled with fundamental geometric objects determining Lagrange--Finsler and/or Einstein spaces. Prescribing corresponding classes of generating
functions, we construct different types of stochastic, fractional, nonholonomic etc models of evolution for nonlinear dynamical systems, exact solutions of Einstein equations and/or Lagrange--Finsler configurations. The main result of this paper consists in a  proof of existence of unique and positive solutions of nonlinear diffusion equations which can be related to stochastic solutions in gravity and Ricci flow theory. This allows us to formulate stochastic modifications of Perelman's functionals and prove the main theorems for stochastic Ricci flow evolution. We show that nonholonomic Ricci flow diffusion  can be with self--organized critical behavior, for gravitational and Lagrange--Finsler systems, and that a
statistical/ thermodynamic analogy to stochastic geometric evolution can be formulated.



{\ \qquad \qquad } 83C15

PACS2008:\ 02.50.Ey, 05.40.-a, 05.65.+b, 04.20.Jb
\end{abstract}

\tableofcontents

\section{Introduction}

The purpose of this paper is to analyze models of Ricci flows of random
metrics with self--organize criticality in general relativity (and various
modifications, for instance, of Lagrange--Finsler type theories of gravity)
in the framework of stochastic evolution equations. In other words, we shall
analyze scenarios of stochastic/diffusion Ricci evolution of curved spaces
with rich geometric/physical structure\footnote{%
when, for instance, metrics are solutions of Einstein equations with certain
prescribed symmetries and/or nonholonomic constraints, solitonic
configurations, stochastic properties etc} resulting in random gravitational
systems with critical points as attractors.

This is a partner work to \cite{vpart1}; our approach to stochastic
solutions and diffusion of gravitational fields is elaborated in the
framework of the theory of stochastic evolution equations following the
methods elaborated in \cite{barbu1,barbu2,barbu3,barbu4} (see also
references therein) but modified/adapted for curved spaces and nonholonomic
Ricci evolution of Riemannian and non--Riemannian geometries, \cite%
{vricci1,vricci2,vricci3,vricci4}).

The theory of stochastic processes and diffusion on curved spaces has been
studied in mathematics and physics from different perspectives related to
new directions in differential geometry and partial differential equations, geometric analysis and evolution theory, kinetic and thermodynamic
processes with local anisotropy, and various applications in cosmology and
astrophysics. We refer to a series of works containing original key ideas,
methods and reviews of results \cite{elw,ikeda,emery,vstoch3,vstoch5,
vstoch6,hu,chev,dunk,bailleul,herrmann,herrmann1}.

In our recent paper \cite{vpart1}, we studied exact solutions in gravity
defined by stochastic generating functions and nonholonomic diffusion of
gravitational fields. Such generic off--diagonal metrics
(which can not be diagonalized by coordinate transforms) describe generic
nonlinear and random processes with very complex structure of gravitational
''spacetime ether'', for instance, solitonic interactions \cite%
{vsolit1,vsolit2}, fractional configurations \cite%
{vfract1,vfract2,vfract3bal} and nonlinear wave interactions and geometric
evolution \cite{vsingl2,vrichiwave1,vrichiwave2}.

In this work, the gravitational stochastic Ricci flow and field interactions
phenomena with self--organized criticality are studied following a synthesis
of the method of anholonomic deformation/frames of constructing exact
solutions in gravity \cite{vsgg,vexsol1,vrevflg,vexsol2,vexsol3} and the
theory of nonholonomic diffusion on curved spaces (additionally to above
mentioned works, we cite our contributions on nonholonomic diffusion in
(super) vector bundles and locally anisotropic kinetics and thermodynamics %
\cite{vstoch3,vstoch4,vstoch5,vstoch6}). To the best our knowledge this is
the first attempt to provide an unified geometric formalism for locally
anisotropic diffusion processes and Ricci flows of Einstein/--Finsler spaces, 
when certain classes of nonholonomic constraints result in stochastic
evolution equations which can be approached in a mathematically strict way %
\cite{barbu1,barbu2,barbu3,barbu4}.\footnote{%
During last two decades, self--organized criticality is widely studied in
physics \cite{jen,caf,car,hwa,jan}. Our idea is that such effects are
possible for gravitational fields and their evolution being derived as
corresponding solutions. Classical vacuum structure in general relativity
may have various sophisticate nontrivial topological and geometric
configurations with possible stochastic, solitonic, instanton, black hole,
wormhole etc properties. We can say that a curved spacetime is modeled by a
nontrivial fundamental gravitational ''ether'' which, as a matter of
principe, may have a very complex nonlinear sure and/or random behavior
which may result in configurations with self--organized criticality. The
term ''porous'' media is a conventional one for gravitational configurations
related to corresponding equations which are formally similar to certain
equations for real porous matter. In this work, we show in explicit form
what type of nonholonomic constraints we have to impose on gravitational
field equations, and their possible geometric evolution, in order to have
critical points as attractors.}

This paper is organized as follow. In section \ref{secrf} we provide an
introduction to the theory of Ricci flow evolution of nonholonomic
geometries. Section \ref{secnd} is a development of  some recent results
on stochastic diffusion equations (in our case) on nonholonomic manifolds
and to proof of existence of unique and positive solutions for such systems
in gravity and geometric mechanics. In section \ref{secsnrf} we develop the
theory of stochastic nonholonomic Ricci flows: there are considered
stochastic modifications of Perelman's functionals and proven the main
theorems for stochastic evolution equations. Finally, a statistical analogy
for stochastic Ricci flows is proposed.

\vskip5pt

\textbf{Acknowledgement: } The author is grateful to Academician V. Barbu
for the opportunity to present a part of results of this paper at the
Seminar of Institute of Mathematics ''O. Mayer' of Romanian Academy, Ia\c{s}%
i, Romania.

\section{Ricci Flow Evolution and Nonholonomic Manifolds}

\label{secrf} The Ricci flow theory \cite{ham1,gper1,gper2,gper3}, related
geometric analysis and various applications (see \cite{caozhu,kleiner,rbook}
for reviews of results and methods) became one of the most intensively
developing branch of modern mathematics. The most important achievement of
this theory was the proof of W. Thurston's Geometrization Conjecture by
Grisha Perelman \cite{gper1,gper2,gper3}. The main results on Ricci flow
evolution were proved originally for (pseudo) Riemannian and K\"{a}hler
geometries. In a series of our works \cite%
{vricci1,vricci2,vricci3,vricci4,vfract1,vrfsol2,vrfsol3}, we studied Ricci
flow evolutions of geometries and physical models (of gravity with symmetric
and nonsymmetric metrics, geometric mechanics, fractional and noncommutative
generalizations) when the field equations are subjected to nonholonomic
constraints and the field/evolution solutions, mutually transform as Riemann
and generalized non--Riemann geometries. In this section we outline the
anholonomic deformation method of constructing exact solutions in gravity
and Ricci flow theory in a form necessary for further developments (in next
sections) for stochastic evolution and Ricci flow diffusion.

\subsection{Modelling Einstein and Finsler geometries on nonholonomic
manifolds}

In the following we shall consider a nonholonomic manifold $\mathbf{V}=(V,%
\mathcal{N})$ enabled with nonlinear connection (N--connection) structure $%
\mathbf{N}: T\mathbf{V}\rightarrow h\mathbf{V}\oplus v\mathbf{V}$ as a
nonholonomic distribution $\mathcal{N}=\mathbf{N}$ defining a conventional
horizontal (h) and vertical (v) splitting on $V.$ \ We consider that $V$ is
a four dimensional, 4--d, (pseudo) Riemann manifold of signature $(+,+,-+).$
For models of Lagrange--Finsler geometry, we can take $\mathbf{V}=TM$ \ to
be the total space of a manifold $M,$ when the typical fiber in such a
manifold is provided with a hyperbolic structure in order to mimic a local
(pseudo) Euclidian structure. We shall use the notations from \cite{vpart1}
(readers may find details on the geometry of nonholonomic manifolds \cite%
{vranceanu1,vranceanu2,vranceanu3}) in monograph \cite{bejf}; for purposes
of modern geometry, mechanics and physics and mathematical relativity, the
formalism is developed in \cite{ma,vrevflg,vsgg}; the concept of
N--connection is contained in coefficient form in \cite{cartan} being
developed in details on a number of works on Lagrange--Finsler geometry and
generalizations/ modifications \cite{bejancu,bcs,vstoch4,vricci2,vricci4}).

We label local coordinates on $\mathbf{V}$ in the form $u^{\alpha
}=(x^{i},y^{a}),$ were $i,j,...=1,2$ and $a,b,...=3,4$ (in brief, we write $%
\mathbf{u=(x,y)}$)$;$ similar values are taken by any variants of primed,
underlined etc indices, for instance, $\underline{\alpha }=(\underline{i},%
\underline{a})$ and $\beta ^{\prime }=\left( j^{\prime },b^{\prime }\right) $%
. A local frame and co--frame are written $e_{\alpha }=(e_{i},e_{a})$ and $%
e^{\beta }=(e^{j},e^{b}),$ when frame transforms are parametrized $e_{\alpha
}=e_{\ \alpha }^{\alpha ^{\prime }}(\mathbf{u})e_{\alpha ^{\prime }}.$

\subsubsection{Nonlinear connections and Einstein equations}

A N--connection structure on $\mathbf{V,}$ with local coefficients $%
\{N_{i}^{a}\}$ stated with respect to a coordinate basis, allows us to
define the so--called N--elongated (equivalently, N--adapted) frames, i.e. ,
respectively, partial derivatives and differentials,
\begin{eqnarray}
\mathbf{e}_{\alpha } &\doteqdot &\left( \mathbf{e}_{i}=\partial
_{i}-N_{i}^{a}\partial _{a},e_{b}=\partial _{b}=\frac{\partial }{\partial
y^{b}}\right) ,  \label{ddr} \\
\mathbf{e}_{\ }^{\beta } &\doteqdot &\left( e^{i}=dx^{i},\mathbf{e}%
^{a}=dy^{a}+N_{i}^{a}dx^{i}\right) .  \label{ddf}
\end{eqnarray}%
With respect to such bases, the geometric objects are written in N--adapted
form and called as distinguished objects (in brief, d--objects), for
instance, d--vectors, d--tensors, d--connections etc.

Our geometric arena consists from nonholonomic manifolds/ bundles
(we shall use as equivalent the terms, spaces, or spacetimes,
 for corresponding signatures) given by
data $\left( \mathbf{N},\mathbf{g},\mathbf{D}\right) ,$ where the d--metric $%
\mathbf{g}$ is parametrized in the form
\begin{equation}
\ \mathbf{g}=g_{ij}{dx^{i}\otimes dx^{j}}%
+h_{ab}(dy^{a}+N_{k}^{a}dx^{k}){\otimes }(dy^{b}+N_{k}^{b}dx^{k}),
\label{dm}
\end{equation}%
and $\mathbf{D}$ may be taken to be the canonical d--connection $\widehat{%
\mathbf{D}}=\left( h\widehat{D},v\widehat{D}\right) $ uniquely defined from
the conditions that it is metric compatible, $\widehat{\mathbf{D}}\mathbf{g}%
=0,$ and with vanishing ''pure'' h-- and v--components of torsion $\widehat{%
\mathbf{T}}$ of $\ \widehat{\mathbf{D}}.$ Even there are nontrivial h--v
components of torsion, there is a distortion relation
\begin{equation}
\widehat{\mathbf{D}}=\nabla +\widehat{\mathbf{Z}},  \label{distcon}
\end{equation}%
where all values, i.e. the \ canonical d--connection, $\widehat{\mathbf{D}}%
=\{\widehat{\mathbf{\Gamma }}_{\ \beta \gamma }^{\alpha }\},$ the
Levi--Civita connection $\nabla =\{\Gamma _{\ \beta \gamma }^{\alpha }\}$
(subjected to the conditions that it is torsionless and $\nabla \mathbf{g}=0$%
) and the distortion, $\widehat{\mathbf{Z}}=\{\widehat{\mathbf{Z}}_{\ \beta
\gamma }^{\alpha }\sim \widehat{\mathbf{T}}_{\ \beta \gamma }^{\alpha }\},$
are completely and uniquely defined by the same metric tensor $\mathbf{g.}$
We can work equivalently with both linear connections $\widehat{\mathbf{D}}$
and/or $\nabla $ (the first one is N--adapted but the second one is not).%
\footnote{%
The formulas for N--adapted coefficients of geometric objects (\ref{distcon}%
) are given in many our works (see, for instance, \cite{vpart1,vrevflg,vsgg}
and \cite{ma}, for explicit constructions in Lagrange--Finsler gravity). For
simplicity, we omit such considerations in this work.}

The Einstein equations in general relativity can be written for $%
\widehat{\mathbf{D}},$
\begin{eqnarray}
&&\widehat{\mathbf{R}}_{\ \beta \delta }-\frac{1}{2}\mathbf{g}_{\beta \delta
}\ ^{s}R=\mathbf{\Upsilon }_{\beta \delta },  \label{cdeinst} \\
&&\widehat{L}_{aj}^{c}=e_{a}(N_{j}^{c}),\ \widehat{C}_{jb}^{i}=0,\ \Omega
_{\ ji}^{a}=0,  \label{lcconstr}
\end{eqnarray}%
where $\widehat{\mathbf{R}}_{\ \beta \delta }$ is the Ricci tensor for $%
\widehat{\mathbf{\Gamma }}_{\ \alpha \beta }^{\gamma },\ ^{s}R=\mathbf{g}%
^{\beta \delta }\widehat{\mathbf{R}}_{\ \beta \delta }$ and the source
d--tensor $\mathbf{\Upsilon }_{\beta \delta }$ is constructed for the same
metric but with $\ \widehat{\mathbf{D}}$ (formulas are similar to those in
general relativity with $\nabla ;$ we consider $\mathbf{\Upsilon }_{\beta
\delta }\rightarrow \varkappa T_{\beta \delta }$ \ for $\widehat{\mathbf{D}}%
\rightarrow \nabla $). If the constraints (\ref{lcconstr}) are satisfied the
distortion d--tensors $Z_{\ \alpha \beta }^{\gamma }$ from (\ref{distcon}),
determined by d--torsion $\widehat{\mathbf{T}}_{\ \alpha \beta }^{\gamma },$
became zero and (\ref{cdeinst}) are equivalent to the ''standard'' equations
\begin{equation}
R_{\ \beta \delta }-\frac{1}{2}g_{\beta \delta }R=\varkappa T_{\beta \delta
},  \label{einsteq}
\end{equation}%
written for the Levi--Civita connection $\nabla =\{\Gamma _{\ \alpha \beta
}^{\gamma }\}.$ In formulas (\ref{einsteq}), $R_{\ \beta \delta }$ and $R$
are respectively the Ricci tensor and scalar curvature of $\nabla ;$ it is
also considered the energy--momentum tensor for matter, $T_{\alpha \beta },$%
where $\varkappa =const.$

Here we note that it is not possible to integrate analytically, in general
form, the Levi--Civita form of the system of partial differential equations (%
\ref{einsteq}) because of it generic nonlinearity and complexity.
Nevertheless, it is possible to solve in very general forms the version for
the canonical d--connection (and for the Cartan d--connection in Finsler
gravity) the system (\ref{cdeinst}), see details in Refs. \cite%
{vexsol2,vexsol3,vexsol1}, which motivates the idea to introduce the
so--called Lagrange--Finsler variables in general relativity. Imposing
additionally the conditions (\ref{lcconstr}), we can restrict the
nonholonomic integral varieties and generate solution of Einstein equations
for the Levi--Civita connection.

\subsubsection{Lagrange--Finsler structures and Einstein gravity}

Let us consider how data $\left( \mathbf{N},\mathbf{g},\mathbf{D}\right) $
can be parametrized on $\mathbf{V}=TM$ in such a form that they will define
a Lagrange and/or Finsler geometry. We introduce frame transforms $\
^{L}g_{\alpha ^{\prime }\beta ^{\prime }}(\mathbf{u})=e_{\ \alpha ^{\prime
}}^{\alpha }(\mathbf{u})e_{\ \beta ^{\prime }}^{\beta }(\mathbf{u})\mathbf{g}%
_{\alpha \beta }(\mathbf{u}),$ when the coefficients of d--metric (\ref{dm})
are transformed into coefficients of
\begin{equation}
\ \ ^{L}\mathbf{g}=\ ^{L}g_{ij}{dx^{i}\otimes dx^{j}}+\
^{L}h_{ab}(dy^{a}+\ ^{L}N_{k}^{a}dx^{k}){\otimes }(dy^{b}+\
^{L}N_{k}^{b}dx^{k}),  \label{sasakilm}
\end{equation}%
where $\ ^{L}g_{ij}\sim \ ^{L}h_{ab}=\frac{1}{2}\frac{\partial
^{2}L(x^{i},y^{c})}{\partial y^{a}\partial y^{b}}$ and $\ ^{L}N_{k}^{a}$ are
respectively the Hessian and canonical N--connection defined by a regular
Lagrangian $L(\mathbf{u})=L(x^{i},y^{c}),$ see details in \cite%
{ma,vrevflg,vsgg} (and, for Ricci flows of Lagrange--Finsler structures, %
\cite{vricci3}). The d--metric (\ref{sasakilm}) is called the Sasaki lift on
$TM$ of a regular Lagrange structure. In arbitrary frames/coordinates, a
Lagrange d--metric $\ ^{L}\mathbf{g}$ is represented in the form $\mathbf{g}$
(\ref{dm}) and, inversely, any $\mathbf{g}$ can be transformed into a $\ ^{L}%
\mathbf{g,}$ into certain Lagrange variables for a correspondingly chosen
distribution/generating Lagrange function $L(\mathbf{u}).$ The formal
constructions can be performed geometrically in similar forms on (pseudo)
Riemannian spacetimes and tangent bundles but have different physical
interpretations (in the first case, we can elaborate analogous
Lagrange--Finsler models of Einstein gravity but in the second case it is
provided a Finsler similar geometrization of regular Lagrange mechanics).

We note that taking $L=F^{2}(x,y),$ where the homogeneous on $y$--variables
real function $F$ is the fundamental/generating Finsler function (see
rigorous mathematical definitions and details in \cite{ma,bejancu,bcs}), we
model a Finsler geometry on a (pseudo) Riemannian manifold or on $\mathbf{V=}%
TM.$ Using the geometric formalism of nonholonomic distributions, associated
N--connections and the geometry of nonholonomic manifolds, we work in a
unified form with all types of (pseudo) Riemann and Lagrange--Finsler
geometries. We distinguish such constructions by additional suppositions on
the structure of manifold/bundle spaces and the type of linear connection, $%
\widehat{\mathbf{D}},\nabla $ or (the Cartan--Finsler d--connection) $\ ^{c}%
\mathbf{D,}$ we chose for the geometric data $\left( \mathbf{N},\mathbf{g},%
\mathbf{D}\right) .$

Finally, we shall say that there are used nonholonomic (Lagrange, or
Finsler, variables) in Einstein gravity on $\mathbf{V}$ if the theory is
defined by any data $\left( \mathbf{N}\sim \ \ ^{L}\mathbf{N};\mathbf{g}\sim
\ ^{L}\mathbf{g};\mathbf{D}\sim \widehat{\mathbf{D}},\mbox{ or }\ \ ^{c}%
\mathbf{D}\right) ,$ which can be equivalently transformed, via distortion
relations of type (\ref{distcon}), into data $\left( \mathbf{g},\nabla
\right) .$ For theories on $TM,$ we can elaborate
Einstein--Lagrange/--Finsler like theories (with metrics and connections
depending on velocity type coordinates, $y^{a}$) using instead of $\nabla $
any convenient for physical purposes $\widehat{\mathbf{D}},$ or $\ ^{c}%
\mathbf{D}$ (or any other metric compatible d--connection completely defined
by $\mathbf{g}$ and adapted to a chosen $\mathbf{N}$). In abstract form, a
N--adapted Einstein model on a general nonholonomic manifold $\mathbf{V}$
generates for corresponding nonholonomic distributions an Einstein and/or
Lagrange--Finsler, spacetime geometry.

\subsection{Nonholonomic Ricci flows of Einstein/--Finsler spaces}

The original Ricci flow theory was proposed \cite{ham1} with an evolution
equation for a set of Riemannian metrics $g_{\alpha \beta }(\chi )$ and
corresponding Ricci tensors $\ R_{\alpha \beta }(\chi )$ parametrized by a
real parameter $\chi .$ We can write the Hamilton's equations in the so--called
normalized form, see details in \cite{caozhu,kleiner,rbook},
\begin{equation}
\frac{\partial }{\partial \chi }g_{\underline{\alpha }\underline{\beta }%
}=-2\ R_{\underline{\alpha }\underline{\beta }}+\frac{2r}{5}Rg_{\underline{%
\alpha }\underline{\beta }},  \label{feq}
\end{equation}%
describing (holonomic) Ricci flows with respect to a coordinate base $%
\partial _{\underline{\alpha }}=\partial /\partial u^{\underline{\alpha }};$
the normalizing factor $r=\int \ RdVol/Vol$ is introduced in order to
preserve the volume $Vol;$ $\ R_{\underline{\alpha }\underline{\beta }}$ and
$\ R=g^{\underline{\alpha }\underline{\beta }}\ R_{\underline{\alpha }%
\underline{\beta }}$ are computed for the Levi--Civita connection $\nabla .$%
\footnote{%
We underline the indices with respect to the coordinate bases but not with
respect to some 'N--elongated' local bases.} Grisha Perelman's fundamental
results \cite{gper1,gper2,gper3} were based on original idea to prove that
the Ricci flow is not only a gradient flow but also can be defined \ as a
dynamical system on the spaces of Riemannian metrics. He introduced two
Lyapunov type functionals and proved that evolution equations of type (\ref%
{feq}) can be derived following a corresponding variational calculus.

In our approach, we studied modified Ricci flow evolution equations, and the
corresponding N--adapted functionals, for generalized commutative and
noncommutative geometries (Riemann and Lagrange--Finsler ones) \cite%
{vricci1,vricci2,vricci3,vricci4,vfract1}. We also were interested to study
Ricci flows of exact solutions in various theories of gravity \cite%
{vrfsol2,vrfsol3,vrichiwave1,vrichiwave2}. \ In this work, we shall
elaborate a theory of Ricci evolution with self--organized criticality for
stochastic solutions for Einstein spaces following geometric using the
canonical d--connection $\widehat{\mathbf{D}}.$

If $\nabla \rightarrow \widehat{\mathbf{D}},$ we have to change $R_{\alpha
\beta }\rightarrow \widehat{\mathbf{R}}_{\alpha \beta }$ which transforms
equations (\ref{feq}) into an N--adapted system of evolution equations for
Ricci flows of (in this work, we consider symmetric metrics, see details in %
\cite{vricci1,vricci2}),
\begin{eqnarray}
\frac{\partial }{\partial \chi }g_{ij} &=&2\left[ N_{i}^{a}N_{j}^{b}\ \left(
\underline{\widehat{R}}_{ab}-\lambda h_{ab}\right) -\underline{\widehat{R}}%
_{ij}+\lambda g_{ij}\right] -h_{cd}\frac{\partial }{\partial \chi }%
(N_{i}^{c}N_{j}^{d}),  \label{e1} \\
\frac{\partial }{\partial \chi }h_{ab} &=&-2\ \left( \underline{\widehat{R}}%
_{ab}-\lambda h_{ab}\right) ,\   \label{e2} \\
\ \widehat{R}_{ia} &=&0\mbox{ and }\ \widehat{R}_{ai}=0,  \label{e3}
\end{eqnarray}%
where the Ricci coefficients $\underline{\widehat{R}}_{ij}$ and $\underline{%
\widehat{R}}_{ab}$ are computed with respect to coordinate coframes and the
cosmological constant $\lambda $ includes the normalization factor $\frac{2r%
}{5}.$ On nonholonomic manifolds, we can prescribe any convenient for our
purposes nonholonomic distributions. For simplicity, we shall fix such
distributions when $\frac{\partial }{\partial \chi }(N_{i}^{c})=0$ $\ $and
model stochastic and/or Ricci flow evolution for nonholonomic Einstein
manifolds with $\widehat{R}_{ab}=\lambda h_{ab}$ and $\widehat{R}%
_{ij}=\lambda g_{ij}$ for so--called ''stationary configurations'', when the
equations (\ref{cdeinst}) with the source determined by $\lambda $ are
satisfied. \ The h-- and v--components of Ricci flow evolutions of
gravitational configurations for $\frac{\partial }{\partial \chi }%
(N_{i}^{c})=0$ , are similar to those derived from Perelman's nonholonomic
functionals\ after corresponding nonholonomic deformations.\footnote{
In next section, we derive such equations for stochastic Ricci flows and
related Perelman's functionals.} Such equation can be written in the form
\begin{eqnarray}
\frac{\partial g_{ij}}{\partial \chi } &=&-2\widehat{R}_{ij},\quad \ \frac{%
\partial h_{ab}}{\partial \chi } =-2\widehat{R}_{ab},\   \label{rf1} \\
\ \widehat{R}_{ia} &=&0\mbox{ and }\ \widehat{R}_{ai}=0.  \notag
\end{eqnarray}

For simplicity, we shall analyze Ricci flows of families of ansatz for
d--metrics parametrized in the form
\begin{eqnarray}
\ _{\chi }\mathbf{g} &\mathbf{=}&e^{\psi (x^{k})}{dx^{i}\otimes dx^{i}}%
+h_{3}(x^{k},t,\chi )\mathbf{e}^{3}{\otimes }\mathbf{e}^{3}+h_{4}(x^{k},t,%
\chi )\mathbf{e}^{4}{\otimes }\mathbf{e}^{4},  \label{genans} \\
\mathbf{e}^{3} &=&dt+w_{i}(x^{k},t,\chi )dx^{i},\mathbf{e}%
^{4}=dy^{4}+n_{i}(x^{k},t,\chi )dx^{i}.  \notag
\end{eqnarray}%
The system of equations (\ref{rf1}) for $\widehat{R}_{ij}=0$ and $\Upsilon
_{2}(x^{k},v)=\lambda $ and existing one Killing symmetry, on vector $%
\partial /\partial y^{4}$ (the coefficients do not depend on variable $y^{4}$%
), with evolution only of the v--parts, transform into
\begin{eqnarray}
&&\ddot{\psi}+\psi ^{\prime \prime }=0,  \label{eq1} \\
&&\frac{\partial }{\partial \chi }h_{3}=-\frac{h_{3}\phi ^{\ast }}{h_{4}},%
\frac{\partial }{\partial \chi }h_{4}=-\frac{h_{4}\phi ^{\ast }}{h_{3}},
\label{eq2} \\
&&\beta w_{i}+\alpha _{i}=0,  \label{eq3} \\
&&n_{i}^{\ast \ast }+\gamma n_{i}^{\ast }=0  \label{eq4}
\end{eqnarray}%
where
\begin{eqnarray}
~\phi (\chi )~ &=&\phi (x^{k},t,\chi )=\ln |\frac{h_{4}^{\ast }}{\sqrt{%
|h_{3}h_{4}|}}|,\   \label{auxphi} \\
\alpha _{i} &=&h_{4}^{\ast }\partial _{i}\phi ,\ \beta =h_{4}^{\ast }\ \phi
^{\ast },\ \gamma =\left( \ln |h_{4}|^{3/2}/|h_{3}|\right) ^{\ast }.  \notag
\end{eqnarray}%
In the above formulas we wrote the partial derivatives in the form $%
a^{\bullet }=\partial a/\partial x^{1},$\ $a^{\prime }=\partial a/\partial
x^{2},$\ $a^{\ast }=\partial a/\partial t$ and we shall use $\partial _{\chi
}a=\partial a/\partial \chi .$ The conditions of zero torsion, i. e.
constraints (\ref{lcconstr}), are written in the form
\begin{equation}
w_{i}^{\ast }=\mathbf{e}_{i}\ln |h_{4}|,\mathbf{e}_{k}w_{i}=\mathbf{e}%
_{i}w_{k},\ n_{i}^{\ast }=0,\ \partial _{i}n_{k}=\partial _{k}n_{i}.
\label{lcconstr1}
\end{equation}%
We have to impose additionally (\ref{lcconstr1}) if we wont to consider
nonholonomic Ricci flows determined by equations (\ref{eq1})--(\ref{eq4}) in
a form when the configurations for the Levi--Civita connections $\nabla
(\chi )$ are extracted.

If $h_{4}^{\ast }(\chi )\neq 0;\Upsilon _{2}=\lambda \neq 0,$ we get $\phi
^{\ast }(\chi )\neq 0$ and we can generate families of exact solutions of (%
\ref{cdeinst}), with diagonal nontrivial source $\mathbf{\Upsilon }_{\delta
}^{\alpha }=diag[\lambda ,\lambda ,0,0],$ if
\begin{equation}
h_{4}^{\ast }(\chi )=2h_{3}(\chi )h_{4}(\chi )\lambda /\phi ^{\ast }(\chi ).
\label{4ep2a}
\end{equation}%
We conclude that the family of ansatz for d--metrics (\ref{genans}) \
subjected to the conditions (\ref{eq1})--(\ref{4ep2a}) define Ricci flow
evolutions on a real parameter $\chi $ of a class of generic off--diagonal
solutions determining correspondingly (non) holonomic Einstein manifolds.
They depend explicitly on the type of families of generating functions $\phi
(\chi ).$ If such functions are random ones subjected to the conditions to
solve certain stochastic/diffusion equations, we can say that we generated
stochastic Einstein equations (for details, see our partner work \cite%
{vpart1} where there are also analyzed classes of solutions of Einstein
equations with $h_{4}^{\ast }=0,$ or $h_{3}^{\ast }=0;$ the length of this
paper does not allow us to analyze the Ricci flow evolution of such more
special classes of \ Einstein manifolds) evolving, in general, randomly, on
parameter $\chi .$

\section{Nonholonomic Diffusion and Self--Organized Cri\-ticality}

\label{secnd} We studied various examples with physically important
solutions (black holes/ellipsoid, wormholes, pp-- and/or solitonic waves
etc) evolving under nonholonomic Ricci flows \cite%
{vrfsol2,vrfsol3,vrichiwave1,vrichiwave2}. Those solutions were with sure
coefficients and sources. To our knowledge, it was not yet analyzed the
evolution of, in general, non--Riemann geometries when certain coefficients
of metrics and connections are random ones, i.e. the Ricci flow theory is
with stochastic evolution.

In this work, we shall develop such a theory of stochastic Ricci flows, for
simplicity, for families of metrics (\ref{genans}), which allows us to
generate solutions in explicit form and to prove the phenomena of
self--organizing criticality of gravitational fields for various stationary
and stochastic Ricci flow evolution. The results from papers \cite%
{barbu1,barbu2,barbu3,barbu4} are crucial in proving the existence of unique
and positive solutions. It should be emphasized here that the formalism of
nonholonomic distributions with associated N--connections is a very
important geometric tool for \ connecting the theory of Ricci flows to
nonlinear diffusion and gravity theories and possible applications.

\subsection{Stochastic diffusion equations on nonholonomic manifolds}

We studied some examples of nonlinear stochastic diffusion equations on
nonholonomic $\mathbf{V}$ in the partner work \cite{vpart1}, for the
so--called N--adapted $\left( \mathbf{L,A}\right) $--diffusion and
stochastic solutions of Einstein equations. It was constructed the
corresponding Laplace--Beltrami operator for the canonical d--connection $%
\widehat{\mathbf{D}}.$ We can generate stochastic d--metrics of type (\ref%
{genans}) if we take random generating functions $\phi (\chi )$ (\ref{auxphi}%
), with stochastic evolution parameter $\chi .$ Using (\ref{4ep2a}), various
families of exact solutions of Einstein equations can be defined. An
important mathematical/physical problem is to state certain general
conditions when the Ricci flow evolution equation have unique and positive
solutions.

Let us introduce of nonholonomic geometric and N--adapted stochastic
calculus framework.  We consider an open bounder domain $\mathcal{U}\subset
\mathbf{V,}$\footnote{%
we suppose that our nonholonomic manifold is covered by such open regions}
with the spatial $\dim \mathbf{V}\leq 3$ with smooth boundary $\partial
\mathcal{U},$ when the Laplace--Beltrami operator $\ \widehat{\bigtriangleup
}$ is completely defined by data $\left( \mathbf{N},\mathbf{g},\widehat{%
\mathbf{D}}\right) .$ Out stochastic nonholonomic Ricci flows will be
modelled using a nonlinear evolution equation%
\begin{eqnarray}
\delta U(\chi )-\widehat{\Delta }\Psi (U(\chi ))\delta \chi &\ni &\sigma
(U(\chi ))\delta \mathcal{W}(\chi ),\mbox{ on }(0,\infty )\times \mathcal{U},
\notag \\
\Psi (U(\chi )) &\ni &0,\mbox{ on }(0,\infty )\times \partial \mathcal{U},
\label{nonlindiff} \\
U(0,u) &=&u\mbox{ on }\mathcal{U}.  \notag
\end{eqnarray}%
In the above formulas, $\delta \mathcal{W}(\chi )$ is a Wiener process, the
initial datum $u$ is given for ''rolling'' the stochastic process on
nonholonomic curved manifold $\mathbf{V},$ locally on carts of a covering
atlas and we can consider any maximal monotone (possible multivalued) graph
with polynomial growth of ''coercive'' function $\Psi :\mathbb{R}\rightarrow
2^{\mathbb{R}}.$ It is also possible to introduce a correspondingly
parametrized random forcing term%
\begin{equation}
\sigma (U)dW=\sum\limits_{k=1}^{\infty }\nu _{k}U\langle l,e_{k}\rangle
_{2}e_{k}  \label{randforc}
\end{equation}%
for any $l\in L^{2}(\mathcal{U}),$ where $\langle \cdot ,\cdot _{k}\rangle
_{2}$ and $e_{k}$ are respectively the scalar product and an orthonormal
basis in $L^{2}(\mathcal{U})$ (which is determined, in our case, for a given
d--metric structure on $\mathbf{V}$), when $\nu _{k}$ is a sequence of
positive number and the set $\beta _{k}=\langle l,e_{k}\rangle $ $\ $can be
associated to a sequence of independent standard Brownian motions on a
filtered probability space $\left( \Omega ,\mathcal{F},\{\mathcal{F}_{\chi
}\}_{\chi \geq 0},\mathbb{P}\right) .$

The physical meaning of equations (\ref{nonlindiff}) depends on the type of
models we are going to elaborate. For instance, such an equation described
the dynamics of flows in porous media; for more general assumptions, it
models phase transitions with melting and solidification processes in the
presence of a random force, i. e. term. Our proposal \cite{vpart1} was to
introduce random generating functions for solutions in gravity \ (with the
possibility to include noise for matter energy--momentum ) considering
vacuum non--vacuum spacetime configurations as a complex ''ether'' media
(following Mach ideas on inertial forces and gravity) with possible
nonholonomic structure, singularities and diffusion of gravitational and
matter fields.

In this work, we show that gravitational fields (curved spacetime) may
self--organize itself due to stochastic gravitational effects and/or under
nonholonomic stochastic Ricci flow evolution. It is possible to formulate
stochastic analogs of gravitational Ricci flow evolution equations (a
geometric nonlinear diffusion driven by with stochastic components of the
Ricci tensor, with certain limits to nonholonomic versions Laplace--Beltrami
operators).

Our goal is to prove (using a synthesis of methods and results from \cite%
{barbu1,barbu2,barbu3,barbu4} and Perelman's functional approach \cite%
{gper1,gper2,gper3} generalized for nonholonomic geometries \cite%
{vricci1,vricci2}) that a class of stochastic evolution/gravitational field
equations can uniquely solved in a form satisfying positive conditions.
There will be used sure evolution analogs of equations (\ref{nonlindiff})
for a function $\widehat{f}(x^{1},x^{2},v,\chi)$ which is modelled by
\begin{equation}
\ \frac{\partial \widehat{U}}{\partial \chi }=-\widehat{\Delta }\widehat{U}%
+\left| \widehat{\mathbf{D}}\widehat{U}\right| ^{2}-R-S,  \label{analogstoch}
\end{equation}%
where $R+S$ is the scalar curvature of $\widehat{\mathbf{D}}.$ Considering
random generating functions we induce a gravitational random forcing term in
formula (\ref{randforc}). We say that $\Psi $ is such way defined that the
terms $-\widehat{\Delta }\widehat{U}+\left| \widehat{\mathbf{D}} \widehat{U}%
\right| ^{2}$ are analogs of $\widehat{\Delta }\Psi $ in (\ref{nonlindiff})
when a unique solution of this equation is related to a unique solution of (%
\ref{analogstoch}) in direct form or computing certain expectation values.
We shall also consider unique stochastic solution for any well defined
functions $\widehat{f}(\widehat{U})=\widehat{f}(u)=\widehat{f}%
(x^{1},x^{2},v,\chi ).$

\subsection{Existence of unique and positive solutions; physical setting for gravity and Lagrange--Finsler spaces}

On flat three dimensional
 real spaces, the existence problem for stochastic equations (%
\ref{nonlindiff}) with additive and multiplicative noise was studied in Ref. %
\cite{barbu2} and, with generalized conditions, in \cite{barbu1,barbu3}. We
shall search for general conditions of existence of equations (\ref%
{nonlindiff}) and (\ref{analogstoch}) on $\mathcal{U}\subset \mathbf{V}$ \
under such assumptions:

\begin{hypothesis}
\begin{enumerate}
\item \label{hypoth}The partition with $\mathcal{U}\subset \mathbf{V}$ and
localization of operators $\widehat{\Delta }$ (the domain of this operator
is $H^{2}(\mathcal{U})\cap H_{0}^{1}(\mathcal{U})),$ see below Notation \ref%
{notati}, and $\widehat{\mathbf{D}}$ are such way parametrized via
nonholonomic distributions that the coercive function $\Psi $ is a maximal
monotone multivalent function from $\mathbb{R}$ into $\mathbb{R},$ when $%
0\in \Psi (0);$

\item $\exists $ $C>0$ and $a\geq 0$ when $\forall \in \mathbb{R}$ we can
write $\sup \{|\theta |:\theta \in \Psi (r)\}\leq C\left( 1+|r|^{a}\right) ;$

\item it is possible to fix the nonholonomic distributions, the canonical
d--connection $\widehat{\mathbf{D}}$ and sequence $\nu _{k}$ in such a form
that locally $\sum\limits_{k=1}^{\infty }\nu _{k}^{2}\lambda
_{k}^{2}<+\infty ,$ for $\lambda _{k}$ being the eigenvalues of the
Laplace--Beltrami d--operator $-\widehat{\Delta }$ on $\mathcal{U}\subset
\mathbf{V}$ with Dirichle boundary conditions on $\partial \mathcal{U}. $
\end{enumerate}
\end{hypothesis}

Roughly speaking, the conditions of above hypothesis are defined to ''roll''
on $\mathcal{U},$ in N--adapted form, for the canonical d--connection $%
\widehat{\mathbf{D}}$ and $\widehat{\Delta },$ the assumptions from
Hypothesis 1.1 in Ref. \cite{barbu1}. Similarly, the equations (\ref%
{nonlindiff}) and further developments in this sections are ''nonholonomic
modifications'' of the results on nonlinear diffusion from the mentioned
papers \cite{barbu1,barbu2,barbu3} but related to stochastic Ricci flows via
profs of existence for (\ref{analogstoch}).\footnote{%
We recommend readers to consult those papers on nonlinear diffusion and
self--organized criticality where basic concepts and veri important results
are stated following a rigorous mathematical formalism. The limits of this
work does not allow us to repeat and develop in a detailed form those
results for h-- and v--splitting, with N--adapted constructions, on $%
\mathcal{U}\subset \mathbf{V.}$}

\begin{notation}
\label{notati}

\begin{enumerate}
\item We have $\sigma (u)\in L_{2}\left( L^{2}(\mathcal{U}),H^{-1}(\mathcal{U%
})\right) $ with corresponding spaces/ conditions for all Hilbert--Schmidt
d--operators from $L^{2}(\mathcal{U})$ into $H^{-1}(\mathcal{U}),$ when
there is the Lipschitz continuity from $H^{-1}(\mathcal{U})$ into $%
L_{2}\left( L^{2}(\mathcal{U}),H^{-1}(\mathcal{U})\right) .$

\item By $L^{p}(\mathcal{U}),p\geq 1,$ we denote the space of $p$%
--integrable functions with norm $|\cdot |_{p}.$ The spaces $H^{k}(\mathcal{U%
})\subset L^{2}(\mathcal{U}),$ ($k=1,2$), are the standard Sobolev spaces on
$\mathcal{U}.$ We consider that $H_{0}^{1}(\mathcal{U})$ is the subspace of $%
H^{1}(\mathcal{U})$ with zero trace on the boundary.

\item Denoting $\mathcal{H}$ as a Hilbert space, for $p,q\in \lbrack
1,+\infty ],$ we write \newline
$L_{W}\left[ \left( 0,\ ^{T}\chi \right) ;L^{p}(\Omega ;\mathcal{H})\right] $
for the space of all $q$--integrable processes $z:[0,\ ^{T}\chi ]\rightarrow
L^{p}(\Omega ;\mathcal{H})$ are adapted to the filtration $\{\mathcal{F}%
_{\chi }\}_{\chi \geq 0}.$ It is also used $C_{W}\left[ \left( 0,\ ^{T}\chi
\right) ;L^{2}(\Omega ;\mathcal{H})\right] $ \ for the space of all $%
\mathcal{H}$--valued N--adapted processes being mean square continuous. $L(%
\mathcal{H})$ denotes the space of bounded linear operators equipped with
the above introduced norm.

\item We model $\mathcal{H}$ as a distribution space $\mathcal{H}=H^{-1}(%
\mathcal{U})=\left( H_{0}^{1}(\mathcal{U})\right) ^{\prime }$ provided with
a N--adapted scalar product and norm defined by d--metric $\mathbf{g}$ and
canonical Laplace--Beltrami d--operator, $\mathbf{A}=$ $-\widehat{\Delta },$%
\begin{equation*}
\langle \ ^{1}f,\ ^{2}f\rangle =\int_{\mathcal{U}}\mathbf{A}^{-1}\ ^{1}f(\xi
)\ ^{2}f(\xi )\sqrt{|\mathbf{g}(\xi )|}\delta \xi ,\mbox{ for }|f|_{-1}=%
\sqrt{\langle \ f,\ f\rangle }.
\end{equation*}
\end{enumerate}
\end{notation}

Under above assumptions, on any $\mathcal{U}\subset \mathbf{V}$ and for a
parameter $\chi ,$ we can use the main result from \cite{barbu1} that if $%
u\in L^{p}(\mathcal{U}),p\geq \max \{2a,4\},$ then there is a unique strong
solution to equation (\ref{nonlindiff}) (any such solution is nonnegative if
the initial data $u$ are also nonnegative).

In our nonholonomic setting for gravity and flows of geometries, we can
solve different physical problems: We do not model usual porous media, but a
gravitational ''ether'' interacting as a solution of Einstein equations
and/or following a scenarios of Ricci flow stochastic evolution. For
instance, there are some important examples of modeling for gravitational
configurations by nonholonomic diffusion equations of type (\ref{nonlindiff}%
):

\begin{example}
\begin{enumerate}
\item For certain constants $\ ^{0}\chi ,\rho ,\ ^{1}\alpha ,\ ^{2}\alpha
\in (0,+\infty ),$ we can model a gravitational thermo--field theory with
heat conduction, or phase transitions in spacetime ''porous foam'' if
\begin{equation*}
\Psi (\chi )=\left\{
\begin{array}{c}
\ ^{1}\alpha (\chi -\ ^{0}\chi ),\mbox{\  for \ }\chi <\ ^{0}\chi ; \\
\lbrack 0,\rho ],\mbox{\  for \ }\chi =\ ^{0}\chi ; \\
\ ^{2}\alpha (\chi -\ ^{0}\chi )+\rho ,\mbox{\  for \ }\chi >\ ^{0}\chi .%
\end{array}%
\right|
\end{equation*}

\item In a spacetime with black holes and singularities, it is important to
consider the nonlinear singular nonholonomic diffusion equation%
\begin{equation*}
\delta U(\chi )-\rho div\left[ \ ^{0}\delta \left( U(\chi )\right) \widehat{%
\mathbf{D}}U(\chi )\right] d\chi =\sigma \left( U(\chi )\right) \delta
\mathcal{W}(\chi ),
\end{equation*}%
for N--adapted divergence $div$, where $\ ^{0}\delta $ is the Dirac measure
concentrated at the origin, which is possible for $\ \Psi (\chi )=\left\{
\begin{array}{c}
\rho \chi /|\chi |,\mbox{\  if \ }\chi \neq 0; \\
\lbrack -1,1],\mbox{\  if \ }\chi =0.%
\end{array}%
\right| $

\item Choosing $\Psi (\chi )=|\chi |^{\alpha }sign\chi $ with $0<\alpha \leq
1,$ for a local diffusion problem with free boundary and a random forcing
term proportional to $U(\chi )-\ ^{c}\chi ,$ when the value $\ ^{c}\chi $
determine the critical point of diffusion process $U(\chi ),$ we get the a
particular nonholonomic diffusion equation%
\begin{equation*}
\delta U(\chi )-\widehat{\Delta }\left( Hev(\chi )+\varkappa \right) \left[
U(\chi )-\ ^{c}u\ \right] d\chi =\sigma \left( U(\chi )-\ ^{c}u\ \right)
\delta \mathcal{W}(\chi ).
\end{equation*}%
In this formula, it is used the Heavside step function $Hev(\chi )=\{0,$ if $%
\ \chi <0;[0,1],$ if $\chi =0;1,$ if $\chi >0\}.$ Certain physical models
for self--organized criticality and their rigorous mathematical study are
given in Refs. \cite{caf,jen,barbu1}. For spacetime evolution, such
criticality can be related with behavior of certain effective physical
constants and/or diffusion gravitational phase transitions subjected to
nonholonomic constraints. One might consider scenarios when the
supercritical region $\left\{ U(\chi )>\ ^{c}u\right\} $ is absorbed
asymptotically during evolution by the critical one $\left\{ U(\chi )=\
^{c}u\right\} .$ Such models seem to have applications in modern quantum
gravity (as alternative to Horava--Lifshitz phase transitions) and or in
cosmology driven by Ricci flow models.
\end{enumerate}
\end{example}

 For the above given examples, we can not apply the general existence
theory of infinite dimensional stochastic equations in Hilbert space
(enabled with nonlinear maximal monotone operators). We have to elaborate a
N--adapted approach.
\begin{definition}
We call a solution of (\ref{nonlindiff}) any $\mathcal{H}$--valued continuous $\mathcal{F}_{\chi }$--adapted and N--adapted process $U(\chi
)=U(\chi ,z),$ for $z\in \mathcal{H}$ , on $[0,\ ^{T}\chi ]$ if $U\in L^{p}\left( \Omega \times (0,\ ^{T}\chi )\times \mathcal{U})\right) \cap L^{2}\left( 0,\ ^{T}\chi ;L^{2}(\Omega ;\mathcal{H})\right) ,p\geq a,$
 and $\exists \eta \in L^{p/a}\left( \Omega \times (0,\ ^{T}\chi )\times
\mathcal{U})\right) $ such that $\mathbb{P}$-a.s. 
\begin{eqnarray*}
&&\langle U(\chi ,z),\ e_{j}\rangle _{2}=
\langle z,\ e_{j}\rangle _{2}+ \\
&&\int\limits_{0}^{\chi }\int\limits_{\mathcal{U}}\eta (s,\xi )\widehat{%
\Delta }e_{j}(\xi )\sqrt{|\mathbf{g}(\xi )|}\delta \xi
ds+\sum\limits_{k=1}^{\infty }\nu _{k}\int\limits_{0}^{\chi }\langle U(s,z)\
e_{k},e_{j}\rangle _{2}d\beta _{k}(s),
\end{eqnarray*}%
$\forall j\in \mathbb{N}$ and $\eta \in \Psi (U)$ a.e. in $\Omega \times
(0,\ ^{T}\chi )\times \mathcal{U}.$
\end{definition}

Finally, we formulate a N--adapted generalization of the existence theorem
(Main Result of \cite{barbu1}):

\begin{theorem}
\label{theorbarbu}For each $z\in L^{p}(\mathcal{U}),p\geq \max \{2a,4\}$ and
conditions of Hypothesis \ref{hypoth}, there is a unique solution $U\in
L_{W}^{\infty }\left[ \left( 0,\ ^{T}\chi \right) ;L^{p}(\Omega ;\mathcal{U})%
\right] ,$ see Notation \ref{notati}, to (\ref{nonlindiff}). If additionally
$z$ is nonnegative a.e. in $\mathcal{U}$ \ then $\mathbb{P}$-a.s. $U(\chi
,z)(\xi )\geq 0,$ for a.e. $(\chi ,\xi )\in (0,\infty )\times \mathcal{U}.$
\end{theorem}

Proof is similar to that provided for flat spaces. We have to ''roll'' on
atlas carts and dub the constructions for $h$-- and $v$--components and
using the Laplace--Beltrami operator for $\widehat{\Delta }$ determined by
the canonical d--connection $\widehat{\mathbf{D}}.$ We omit such technical
results. The most important consequence of this Theorem is that the
existence of a unique solution for such diffusion equations can be related to  stochastic Ricci
flows. This allows us to derive in a unique form the evolution
equations and related (stochastic) fundamental functionals.

\section{Stochastic Nonholonomic Ricci Flows}

\label{secsnrf} For a general random generating function $\phi (\chi )$
introduced into a d--metric (\ref{genans}), it is not clear how to define
the Perelman functionals and derive the Hamilton evolution equations for
stochastic Ricci flows. The goal of this section is to sketch in brief a
self--consistent stochastic version of Hamilton--Perelman theory for
nonlinear diffusion of metric coefficients when the assumptions of
Hypothesis \ref{hypoth} and conditions of Theorem \ref{theorbarbu} are
satisfied for certain generating/normalizing functions.

\subsection{A stochastic modification of Perelman's functionals}

The Perelman's functionals were introduced for Ricci flows of Riemannian
metrics and Levi--Civita connection \cite{gper1,gper2,gper3} are written in
the form
\begin{eqnarray}
\ _{\shortmid }\mathcal{F}(\mathbf{g},f) &=&\int\limits_{\mathbf{V}}\left( \
_{\shortmid }R+\left| \nabla f\right| ^{2}\right) e^{-f}\ dV,  \label{pfrs}
\\
\ _{\shortmid }\mathcal{W}(\mathbf{g},f,\tau ) &=&\int\limits_{\mathbf{V}}\left[ \tau
\left( \ _{\shortmid }R+\left| \nabla f\right| \right) ^{2}+f-2n\right] \mu
\ dV,  \notag
\end{eqnarray}%
where $dV(\xi )=\sqrt{|\mathbf{g}(\xi )|}\delta \xi $ is the volume form,
integration is taken over compact $\mathbf{V}$ and $\ _{\shortmid }R$ is the
scalar curvature computed for $\nabla .$ For a real evolution parameter $\tau
>0,$ it is considered $\int\nolimits_{\mathbf{V}}\mu dV=1$ when $\mu =\left(
4\pi \tau \right) ^{-n}e^{-f}.$ The functional approach can be redefined for
N--anholonomic manifolds with stochastic generating functions (in this \
section, we shall follow the N--adapted geometric formalism elaborated in
Refs. \cite{vricci1,vricci3}):

\begin{claim}
\textbf{--Definition:} For nonholonomic manifolds of even dimension $2n$ with stochastically generated geometric
objects, the stochastic Perelman's functionals for the
canonical d--connection $\widehat{\mathbf{D}}$ are
defined%
\begin{eqnarray}
\widehat{\mathcal{F}}(\mathbf{g},\widehat{f}) &=&\int\limits_{\mathbf{V}}\left(
R+S+\left| \widehat{\mathbf{D}}\widehat{f}\right| ^{2}\right) e^{-\widehat{f}%
}\ dV,  \label{npf1} \\
\widehat{\mathcal{W}}(\mathbf{g},\widehat{f},\tau )
&=&\int\limits_{\mathbf{V}}\left[
\widehat{\tau }\left( R+S+\left| ^{h}D\widehat{f}\right| +\left| ^{v}D%
\widehat{f}\right| \right) ^{2}+\widehat{f}-2n\right] \widehat{\mu }\ dV,
\label{npf2}
\end{eqnarray}%
where $dV$ is the volume form of $\ ^{L}\mathbf{g,}$ $R$ and $S$ are
respectively the h- and v--components of the curvature scalar of $\ \widehat{%
\mathbf{D}},$ for $\ \widehat{\mathbf{D}}_{\alpha }=(D_{i},D_{a}),$ or $%
\widehat{\mathbf{D}}=(\ ^{h}D,\ ^{v}D),$ $\left| \widehat{\mathbf{D}}%
\widehat{f}\right| ^{2}=\left| ^{h}D\widehat{f}\right| ^{2}+\left| ^{v}D%
\widehat{f}\right| ^{2},$ and $\widehat{f}$ satisfies $\int\nolimits_{%
\mathbf{V}}\widehat{\mu }dV=1$ for $\widehat{\mu }=\left( 4\pi \tau \right)^{-n}e^{-\widehat{f}}$ and $\tau >0.$
\end{claim}

\begin{proof}
The formulas (\ref{pfrs}) are redefined for some $\widehat{f}$ and $f$ \
(which can be a non--explicit relation between random functions etc and with
gravitational ''noise'' induced by stochastic N--adapted $R+S$) when
\begin{equation*}
\left( \ _{\shortmid }R+\left| \nabla f\right| ^{2}\right) e^{-f}=\left(
R+S+\left| ^{h}D\widehat{f}\right| ^{2}+\left| ^{v}D\widehat{f}\right|
^{2}\right) e^{-\widehat{f}}\ +q.
\end{equation*}%
Re--scaling the evolution parameter $\tau $ (it is similar to the considered
above $\chi $),  $\tau \rightarrow \widehat{\tau },$ we have
{\small
\begin{equation*}
\left[ \tau \left( \ _{\shortmid }R+\left| \nabla f\right| \right) ^{2}+f-2n)%
\right] \mu =\left[ \widehat{\tau }\left( R+S+\left| ^{h}D\widehat{f}\right|
+\left| ^{v}D\widehat{f}\right| \right) ^{2}+\widehat{f}-2n\right] \widehat{%
\mu }+q_{1}
\end{equation*}
}%
for some $q$ and $q_{1}$ for which $\int\limits_{\mathbf{V}}qdV=0$ and $%
\int\limits_{\mathbf{V}}q_{1}dV=0.$ $\square $
\end{proof}

The geometric objects defining functionals $\widehat{\mathcal{F}}$ and $%
\widehat{\mathcal{W}}$ are with some components computed as expectation
values  (see details in \cite{vpart1}, for instance,
formulas (12) - (14) when the probability density function subjected to
Focker--Plank conditions, is used for computing such values). We consider
the h--variation $^{h}\delta g_{ij}=v_{ij},$ the v--variation $^{v}\delta
g_{ab}=v_{ab},$ and $^{h}\delta \widehat{f}=\ ^{h}f,$ $^{v}\delta \widehat{f}%
=\ ^{v}f$

An explicit calculus for the first N--adapted variations of (\ref{npf1}),
see similar details in \cite{vricci1,vricci3}, distinguished into sure and
random components, can be represented in the form
\begin{eqnarray}
&&\delta \widehat{\mathcal{F}}(v_{ij},v_{ab},\ ^{h}f,\ ^{v}f)=  \label{vnpf1}
\\
&&\int\limits_{\mathbf{V}}\{[-v_{ij}(R_{ij}+D_{i}D_{j}\widehat{f})+(\frac{\
^{h}v}{2}-\ ^{h}f)\left( 2\ ^{h}\Delta \widehat{f}-|\ ^{h}D\ \widehat{f}%
|\right) +R]  \notag \\
&&+[-v_{ab}(R_{ab}+D_{a}D_{b}\widehat{f})+(\frac{\ ^{v}v}{2}-\ ^{v}f)\left(
2\ ^{v}\Delta \widehat{f}-|\ ^{v}D\ \widehat{f}|\right) +S]\}e^{-\widehat{f}%
}dV,  \notag
\end{eqnarray}%
where $^{h}\Delta =D_{i}D^{i}$ and $^{v}\Delta =D_{a}D^{a},\widehat{\Delta }%
= $ $\ ^{h}\Delta +\ ^{v}\Delta ,$ and $\ ^{h}v=g^{ij}v_{ij},\
^{v}v=h^{ab}v_{ab}.$

\subsection{Main Theorems for stochastic Ricci flow  equations}

We shall prove that stochastic equations of type (\ref{rf1}) can be derived
from the Perelman's N--adapted functionals (\ref{npf1}) and (\ref{npf2})
(for simplicity, we shall not consider the normalized term and put $\lambda
=0).$

\begin{definition}
A metric $\ \mathbf{g}$ generated by a stochastic generating function $\phi
(\chi )$ is called a (nonholonomic) stochastic breather if for some $\chi
_{1}<\chi _{2}$ and $\alpha >0$ the metrics $\alpha \ \mathbf{g(}\chi _{1}%
\mathbf{)}$ and $\alpha \ \mathbf{g(}\chi _{2}\mathbf{)}$ differ only by a \
N--adapted diffeomorphism. The cases $\alpha =,<,>1$ define correspondingly
the steady, shrinking and expanding breathers (there are configurations
when, for instance, the h--component of metric is steady but the
v--component is shrinking).
\end{definition}

Clearly, the breather properties, in sure and stochastic variables, depend
on the type of connections which are used for definition of Ricci flows. We
can elaborate a unique nonlinear diffusion scenarios for $\widehat{\mathbf{D}%
}$ for the assumptions of Hypothesis \ref{hypoth}.

Following a N--adapted variational calculus for $\widehat{\mathcal{F}}(\mathbf{g},%
\widehat{f}),$ with formula (\ref{vnpf1}), Laplacian $\widehat{\Delta }$ and
h- and v--components of the Ricci tensor, $\widehat{R}_{ij}$ and $\widehat{S}%
_{ij},$ and considering parameter $\tau (\chi ),$ $\partial \tau /\partial
\chi =-1,$ we formulate

\begin{theorem} \label{firstperth}
The stochastic nonholonomic Ricci flows can be
parametrized by corresponding nonholonomic distributions and characterized by evolution equations
\begin{eqnarray}
\frac{\partial g_{ij}}{\partial \chi } &=&-2\widehat{R}_{ij},\ \frac{%
\partial h_{ab}}{\partial \chi }=-2\widehat{R}_{ab},  \notag \\
\ \frac{\partial \widehat{f}}{\partial \chi } &=&-\widehat{\Delta }\widehat{f%
}+\left| \widehat{\mathbf{D}}\widehat{f}\right| ^{2}-R-S  \label{scaleq}
\end{eqnarray}%
and the property that
\begin{equation*}
\frac{\partial }{\partial \chi }\widehat{\mathcal{F}}(\ \mathbf{g(\chi ),}%
\widehat{f}(\chi ))=2\int\limits_{\mathbf{V}}\left[ |\widehat{R}%
_{ij}+D_{i}D_{j}\widehat{f}|^{2}+|\widehat{R}_{ab}+D_{a}D_{b}\widehat{f}|^{2}%
\right] e^{-\widehat{f}}dV,
\end{equation*}%
$\int\limits_{\mathbf{V}}e^{-\widehat{f}}dV$ is constant and the geometric
objects \ $\widehat{\mathbf{D}},\widehat{\Delta },\widehat{R}_{ij},\widehat{R%
}_{ab},R$ and $\ S$ are induced from random generating functions in \ $%
\mathbf{g(}\chi )$.
\end{theorem}

\begin{proof}
We sketch the idea of such a proof: We should follow G. Perelman \cite{gper1}
constructions (details are given for the connection $\nabla $ in the
Proposition 1.5.3 of \cite{caozhu}) but for the canonical d--connection $%
\widehat{\mathbf{D}}$). \ The random components of metrics and connections
are included as mathematical expectations. The most important task is to
prove that for stochastic generating functions the equation (\ref{scaleq})
has a unique solution for well defined conditions. \ Really, this equation is
equivalent to (\ref{analogstoch}) for with a stochastic system (\ref%
{nonlindiff}) can be associated. The existence of a unique solution follows
from Theorem \ref{theorbarbu}. $\square $
\end{proof}

We can also derive stochastic evolution equations, following stochastic
N--adapted modifications of Proposition 1.5.8 in \cite{caozhu} containing
the details of the original result from \cite{gper1}:

\begin{theorem}
\label{theveq}If a d--metric $\ \mathbf{g}(\chi )$ determined by stochastic
generating function $\phi (\chi )$ (\ref{auxphi}) and functions $\widehat{f}%
(\chi )$ and $\widehat{\tau }(\chi )$ evolve stochastically following the
equations (a stochastic nonholonomic generalization of Hamilton's equations)
\begin{eqnarray}
\frac{\partial g_{ij}}{\partial \chi } &=&-2\widehat{R}_{ij},\ \frac{%
\partial h_{ab}}{\partial \chi }=-2\widehat{R}_{ab},  \notag \\
\ \frac{\partial \widehat{f}}{\partial \chi } &=&-\widehat{\Delta }\widehat{f%
}+\left| \widehat{\mathbf{D}}\widehat{f}\right| ^{2}-R-S+\frac{n}{\widehat{%
\tau }},  \label{auxpar1} \\
\frac{\partial \widehat{\tau }}{\partial \chi } &=&-1  \label{auxpr2}
\end{eqnarray}%
and the property that that the stochastic Ricci flow \ evolution is
constrained to satisfy, for $\int\limits_{\mathbf{V}}(4\pi \widehat{\tau }%
)^{-n}e^{-\widehat{f}}dV=const,$  the conditions
\begin{eqnarray*}
&&\frac{\partial }{\partial \chi }\widehat{\mathcal{W}}(\ \mathbf{g}(\chi )%
\mathbf{,}\widehat{f}(\chi ),\widehat{\tau }(\chi ))= \\
&&2\int\limits_{\mathbf{V}%
}\widehat{\tau }[|\widehat{R}_{ij}+D_{i}D_{j}\widehat{f}-\frac{1}{2\widehat{%
\tau }}g_{ij}|^{2}+|\widehat{R}_{ab}+D_{a}D_{b}\widehat{f}-\frac{1}{2%
\widehat{\tau }}g_{ab}|^{2}](4\pi \widehat{\tau })^{-n}e^{-\widehat{f}}dV.
\end{eqnarray*}
\end{theorem}

The equation (\ref{auxpar1}) is similar to (\ref{analogstoch}) and can be
related to a stochastic system of type (\ref{nonlindiff}). The condition (%
\ref{auxpr2}) \ states a relation between a Ricci flow evolution parameter $%
\widehat{\tau }$ and stochastic evolution parameter $\chi .$

\subsection{Statistical analogy for stochastic Ricci flows}

The Ricci flow theory has, in its turn, a very interesting application in
the theory of stochastic equations and diffusion. It allows us to
characterize additionally such processes via associated statistical
functionals, entropy and thermodynamical values. G. Perelman emphasized \cite%
{gper1} that the functional $\ _{\shortmid }\mathcal{W}$ is in a\ sense
analogous to minus entropy. In this section, we prove that such a property
exists also for nonholonomic stochastic Ricci flows and that we can provide
a statistical model for nonlinear diffusion processes.\footnote{%
Let us remember some important concepts from statistical mechanics: The
partition function $Z=\int \exp (-\beta E)d\omega (E)$ for the canonical
ensemble at temperature $\ ^{T}\beta ^{-1}$ is defined by the measure taken
to be the density of states $\omega (E).$ The thermodynamical values are
computed in the form: the average energy, $\langle E\rangle =-\partial \log
Z/\partial \ ^{T}\beta ,$ the entropy $S=\ ^{T}\beta \langle E\rangle +\log Z
$ and the fluctuation $\sigma =\langle \left( E-\langle E\rangle \right)
^{2}\rangle =\partial ^{2}\log Z/\partial \ ^{T}\beta ^{2}.$}

We consider a evolution of stochastic geometric systems described by some
metrics $\mathbf{g}(\widehat{\tau }),$ N--connections $N_{i}^{a}(\widehat{%
\tau })$ and related canonical d--connections $\widehat{\mathbf{D}}(\widehat{%
\tau })$ when the conditions of Theorem \ref{theveq} are satisfied. It
follows:

\begin{theorem}
Any family of stochastically Ricci flow evolving nonholonomic geometries,
and solutions of Einstein equations, with nonlinear diffusion data derived
from Hypothesis \ref{hypoth} and conditions of Theorem \ref{theorbarbu} \ is
characterized by thermodynamic values
\begin{eqnarray}
&&\langle \widehat{E}\rangle  =-\widehat{\tau }^{2}\int\limits_{\mathbf{V}%
}\left( R+S+\left| ^{h}D\widehat{f}\right| ^{2}+\left| ^{v}D\widehat{f}%
\right| ^{2}-\frac{n}{\widehat{\tau }}\right) \widehat{\mu }\ dV,
\label{thermodv} \\
&&\widehat{S} =-\int\limits_{\mathbf{V}}\left[ \widehat{\tau }\left(
R+S+\left| ^{h}D\widehat{f}\right| ^{2}+\left| ^{v}D\widehat{f}\right|
^{2}\right) +\widehat{f}-2n\right] \widehat{\mu }\ dV,  \notag \\
&&\widehat{\sigma } =2\ \widehat{\tau }^{4}\int\limits_{\mathbf{V}}\left[ |%
\widehat{R}_{ij}+D_{i}D_{j}\widehat{f}-\frac{1}{2\widehat{\tau }}%
g_{ij}|^{2}+|\widehat{R}_{ab}+D_{a}D_{b}\widehat{f}-\frac{1}{2\widehat{\tau }%
}g_{ab}|^{2}\right] \widehat{\mu }\ dV.  \notag
\end{eqnarray}
\end{theorem}

\begin{proof}
It follows from a straightforward computation for \newline
$\widehat{Z}=\exp \{\int\nolimits_{\mathbf{V}}[-\widehat{f}+n]\widehat{\mu }%
dV\}$  as in the original paper \cite{gper1}. For nonholonomic stochastic
processes, we have to N--adapt the constructions as in \cite{vricci1,vricci3}%
. The stochastic terms are included in formulas (\ref{thermodv}) via
expected values of smooth coefficients satisfying the Fokker--Plank equation
(or forward Kolmogorov equation). $\square $
\end{proof}

Any N--adapted stochastic configuration determined by a canonical
d--connection $\widehat{\mathbf{D}}$ is thermodynamically more (less,
equivalent) convenient than a similar one defined by the Levi--Civita
connection $\nabla $ if $\ \widehat{S}<\ _{\shortmid }S$ ($\widehat{S}>\
_{\shortmid }S,\widehat{S}=\ _{\shortmid }S$). Similarly, a certain geometry
can be more (less, equivalent) convenient than a stochastic one and related
Ricci flow evolution models.

\end{document}